\documentclass[10pt,conference]{IEEEtran}

\usepackage[ansinew]{inputenc}
\usepackage{graphicx}
\usepackage{epstopdf}
\usepackage{color}
\usepackage{subfigure}

\newcommand{\prob}{{\cal P}}

        \newtheorem{definition}{Definition}%
        \newtheorem{theorem}{Theorem}%
        \newtheorem{lemma}{Lemma}%

\DeclareOldFontCommand{\rm}{\normalfont\rmfamily}{\mathrm}
\DeclareOldFontCommand{\sf}{\normalfont\sffamily}{\mathsf}
\DeclareOldFontCommand{\tt}{\normalfont\ttfamily}{\mathtt}
\DeclareOldFontCommand{\bf}{\normalfont\bfseries}{\mathbf}
\DeclareOldFontCommand{\it}{\normalfont\itshape}{\mathit}
\DeclareOldFontCommand{\sl}{\normalfont\slshape}{\@nomath\sl}
\DeclareOldFontCommand{\sc}{\normalfont\scshape}{\@nomath\sc}

\newtheorem{example}{Example}

\newcommand{\comment}[1]{}

\newcommand{\bml}[1]{\begin{multline}\label{#1}}
\newcommand{\eml}{\end{multline}}
\newcommand{\beq}[1]{\begin{equation}\label{#1}}
\newcommand{\eeq}{\end{equation}}

\newcommand{\beann}{\begin{eqnarray*}}
\newcommand{\eeann}{\end{eqnarray*}}
\newcommand{\bea}[1]{\begin{eqnarray}\label{#1}}
\newcommand{\eea}{\end{eqnarray}}
\newcommand{\bmp}{\begin{minipage}}
\newcommand{\emp}{\end{minipage}}

\newcommand{\eqref}[1]{(\ref{#1})}

\newcommand{\definref}[1]{{\itshape Definition~\ref{#1}}}
\newcommand{\theorref}[1]{{\itshape Theorem~\ref{#1}}}

\newcommand{\secref}[1]{Section~\ref{#1}}

\newcommand{\fig}[1]{{Fig.~\ref{#1}}}

\newcommand{\exref}[1]{{\itshape Example~\ref{#1}}}
\def\SET0N {I\hspace{-0.8ex}N_0}

{%
}%

{%
}%
\newsavebox{\Citname}

\newcommand{\ignore}[1]{}

\begin{document}

\title{One-Shot Capacity of Discrete Channels}

\author{
\authorblockN{Rui A. Costa}
\authorblockA{Instituto de Telecomunica\c{c}\~oes\\
Faculdade de Ci\^encias da\\Universidade do Porto, Portugal\\
Email: rfcosta@fe.up.pt}
\and
\authorblockN{Michael Langberg}
\authorblockA{Computer Science Division\\
Open University of Israel\\
Email: mikel@openu.ac.il}
\and
\authorblockN{Jo\~ao Barros}
\authorblockA{Instituto de Telecomunica\c{c}\~oes\\
Faculdade de Engenharia da\\Universidade do Porto, Portugal\\
Email: jbarros@fe.up.pt}
}

\maketitle

\begin{abstract}
Shannon defined channel capacity as the highest rate at which there exists a sequence of codes of block length $n$ such that the error probability goes to zero as $n$ goes to infinity. In this definition, it is implicit that the block length, which can be viewed as the number of available channel uses, is unlimited. This is not the case when the transmission power must be concentrated on a single transmission, most notably in military scenarios with adversarial conditions or delay-tolerant networks with random short encounters. A natural question arises: how much information can we transmit in a single use of the channel? We give a precise characterization of the one-shot capacity of discrete channels, defined as the maximum number of bits that can be transmitted in a single use of a channel with an error probability that does not exceed a prescribed value. This capacity definition is shown to be useful and significantly different from the zero-error problem statement.
\end{abstract}

\section{Introduction}

Shannon's notion of channel capacity~\cite{shannon:48} is asymptotic in the sense that the number of channels uses (or, equivalently, the block length of the code) can be arbitrarily large. The rate of a code is defined as the ratio between the number of input symbols and the number of channel uses required to transmit them. A rate is said to be achievable if there exists a sequence of codes (of that rate) with block length $n$, whose error probability goes to zero asymptotically as $n$ goes to infinity. The behavior of the channel capacity with a limited number of channel uses is less well understood. A typical approach is to consider the rate at which the error probability decays to zero, which motivates the study of {\it error exponents} \cite{shannon:67},\cite{gallager:65},\cite{korner:82}.
Beyond the aforementioned definitions, it is also reasonable to ask what rates can be achieved when the error probability must be precisely zero. In \cite{shannon:56}, Shannon assumes once again that the channel is available as many times as necessary and defines the zero-error capacity as the supremum of the independence numbers of the extensions of the confusion graph~\cite{lovasz:79,orlitsky:98}.

A different question is how much information can we convey in a single use of the channel or, in other words, what is the {\it one-shot capacity} of the channel. The question arises for example when the transmission power must be concentrated on a single transmission, most notably in military scenarios with adversarial conditions or delay-tolerant networks with random short encounters. As we have seen, classical definitions of capacity do not encompass these scenarios.

The one-shot capacity problem can also be viewed as a special instance of the single-letter coding problem since, in both problems, the encoder must assign to every source output symbol one channel input symbol (and not a sequence of them). However, to the best of our knowledge, studies on single-letter coding use optimality criteria based on an unlimited number of channel uses. For instance, \cite{gastpar:03} characterizes optimal single-letter source-channel codes, with respect to the rate-distortion and capacity-cost functions, which are of asymptotic nature. For comparison, in our study of the one-shot capacity, only a single channel use is considered.

Using a combinatorial approach, the zero-error one-shot capacity of a given channel was considered in \cite{lovasz:79,orlitsky:98}. More specifically, \cite{lovasz:79,orlitsky:98} construct a certain undirected graph $G$ corresponding to the channel at hand; and characterize the zero-error one-shot capacity by the size of the maximum independent set in $G$.

In this work we generalize the results and combinatorial framework of \cite{lovasz:79,orlitsky:98} to capture communication that allows an error probability below $\epsilon$; namely, we study the {\em $\epsilon$-error one-shot} capacity.  We note that preliminary results on the $\epsilon$-error one-shot capacity appear in~\cite{renner:06}, which uses smooth min-entropy and a probabilistic approach to develop bounds for the one-shot capacity (also called, single-serving channel capacity). Our work differs from~\cite{renner:06} in that we characterize the exact value of the $\epsilon$-error one-shot capacity by means of classical combinatorics.

Our main contributions are as follows:
\begin{itemize}
\item {\it Problem Formulation:} We provide a rigorous mathematical framework for analyzing the $\epsilon$-capacity of discrete channels subject to a one-shot constraint. We consider two different metrics of performance: maximum error probability and average error probability.
\item {\it Operational Interpretation:} We illustrate the practical relevance of the one-shot capacity by means of examples where the zero-error one-shot capacity and the $\epsilon$-error one-shot capacity present significantly distinct behaviors.
\item {\it Combinatorial description of the One-Shot Capacity of Discrete Channels:} We cast the capacity in terms of the properties of a special graph $G$ derived from the channel. For maximum error, we describe the one-shot capacity through the independence number of $G$, whereas for average error we consider the maximum size of {\em sparse} sets in $G$.
\item {\it Complexity Analysis:} We show that the problem of computing the one-shot capacity is NP-Hard.
\end{itemize}

The remainder of the paper is organized as follows. In \secref{sec:problem}, we give a formal definition of the problem at hand, namely the concepts of $\epsilon$-maximum and $\epsilon$-average one-shot capacity. In \secref{sec:examples2} we present a non-trivial example of a class of channels for which the one-shot capacity is relevant. Our main result for maximum error one-shot capacity is stated and proved in \secref{sec:main}. In \secref{sec:nphard}, we prove that computing the $\epsilon$-maximum one-shot capacity is NP-Hard. Finally, in \secref{sec:average} we discuss the case of $\epsilon$-average one-shot capacity and \secref{sec:conclu} concludes the paper.

\section{Problem Statement}
\label{sec:problem}

We start our problem statement with the usual definition of a discrete channel.

\begin{definition}
A discrete {\it channel} is composed of an input alphabet $\cal{X}$, an output alphabet $\cal{Y}$ and the transition probabilities $\prob(Y=y|X=x)$.
\end{definition}

We will refer to such a channel as ``the channel described by $P_{Y|X}$".
Next, we present the definition of a one-shot communication scheme over a discrete channel.

\begin{definition}
A {\it one-shot communication scheme} over a $P_{Y|X}$ channel is composed of a codebook $\underline{\cal{X}} \subseteq \cal{X}$, and a decoding function $\gamma:\cal{Y}\rightarrow \underline{\cal{X}}$.
\end{definition}

We will refer to such a communication scheme as the ``$(\underline{\cal{X}},\gamma)$ pair". It is natural to view the set $\underline{\cal{X}}$ as the set of messages to be transmitted over the channel. Our figure of merit is the probability of error in the decoding process. We consider two different metrics: maximum and average error probability.

\begin{definition}
The {\it maximum error probability} associated with a pair $(\underline{\cal{X}},\gamma)$ is defined as \vspace{-0.25cm}$$\epsilon_{\underline{\cal{X}},\gamma}=\max_{x\in \underline{\cal{X}}} \prob(\gamma(Y)\neq x |X=x).\vspace{-0.15cm}$$
\end{definition}

\begin{definition}
The {\it average error probability} associated with a pair $(\underline{\cal{X}},\gamma)$ is defined as \vspace{-0.15cm}$$\bar{\epsilon}_{\underline{\cal{X}},\gamma}=\frac{1}{|\underline{\cal{X}}|}\sum_{x \in \underline{\cal{X}}} \prob(\gamma(Y)\neq x |X=x).\vspace{-0.15cm}$$ 
\end{definition}

We are now ready to define the one-shot capacity of a discrete channel. From an intuitive point of view, we are intrigued by the maximum number of distinct messages (the size of the codebook $\underline{\cal{X}}$) that can be transmitted in a single use of the channel, while ensuring that the error probability (maximum or average) does not exceed a prescribed value $\epsilon$. We must first define an admissible $(\underline{\cal{X}},\gamma)$ pair.

\begin{definition}
Consider a channel described by $P_{Y|X}$ and let $\epsilon \in [0,1]$. The pair $(\underline{\cal{X}},\gamma)$ is {\it maximum-$\epsilon$-admissible} if $\epsilon_{\underline{\cal{X}},\gamma}\leq \epsilon.$ The set of all $\epsilon$-admissible pairs is denoted by $\cal{A}_\epsilon$.
\end{definition}

\begin{definition}
Consider a channel described by $P_{Y|X}$ and let $\epsilon \in [0,1]$. The pair $(P_X,\gamma)$ is {\it average-$\epsilon$-admissible} if $\bar{\epsilon}_{\underline{\cal{X}},\gamma} \leq \epsilon$. The set of all average-$\epsilon$-admissible pairs is denoted by $\bar{\cal{A}_\epsilon}$.
\end{definition}

The notion of single-serving capacity is outlined in~\cite{renner:06} as ``{\it
the maximum number of bits that can be transmitted in a single use of $P_{Y|X}$, such that every symbol can be decoded by an error of at most $\epsilon$}". We formalize this notion by defining the $\epsilon$-maximum one-shot capacity as follows:

\begin{definition}
Consider a channel described by $P_{Y|X}$. For $\epsilon \in [0,1]$, the {\it $\epsilon$-maximum one-shot channel capacity}, denoted by $C_{\epsilon}$, is defined as \vspace{-0.3cm} $$C_{\epsilon}=\max_{(\underline{\cal{X}},\gamma) \in \cal{A}_\epsilon} \log(|\underline{\cal{X}}|).\footnote{Throughout this work, all the logarithms are considered in base 2.}$$
\end{definition}

Similarly, we can define the $\epsilon$-average one-shot capacity as follows:

\begin{definition}
\label{def:average}
Consider a channel described by $P_{Y|X}$. For $\epsilon \in [0,1]$, the {\it $\epsilon$-average one-shot channel capacity}, denoted by $C_{\epsilon}$, is defined as \vspace{-0.25cm} $$\bar{C}_{\epsilon}= \max_{(\underline{\cal{X}},\gamma) \in\bar{\cal{A}_\epsilon}} \log(|\underline{\cal{X}}|) $$
\end{definition}

Our goal is to provide a precise characterization of the one-shot capacity.

\section{Practical Relevance of the One-Shot Capacity}
\label{sec:examples2}

So far, we have formally defined the concept of the $\epsilon$-error one-shot capacity. One question that naturally arises is the following: does the $\epsilon$-error one-shot capacity significantly differ from the zero-error one-shot capacity, for small values of $\epsilon$? In other words, are there classes of channels for which allowing a small error probability enables the transmission of a significantly larger number of bits than in the zero-error case? In this section, we present a class of channels for which the answer to the previous questions is yes, which asserts for the practical relevance of the $\epsilon$-error one-shot capacity notion. Our examples use the maximum error criterion (and thus imply the gap for average error also).

\begin{example}
\label{ex:first}
Consider a channel with input alphabet $\cal{X}=$$\{0,1,\dots,n-1\}$, output alphabet $\cal{Y}=$$\{0,1,\dots,n-1\}$, and let $0< e_1<e_2<\dots <e_{n-1}\leq 1$. Let $\prob(Y=0|X=0)=1$. For each $i\in \textrm{$\cal{X}$}\setminus \{0\}$, let the transition probability distribution be
\vspace{-0.3cm}
\begin{displaymath}
P(Y=y|X=i) = \left\{ \begin{array}{ll}
 1-e_i & \textrm{if $y=i$}\\
 e_i & \textrm{if $y=0$}\\
 0 & \textrm{otherwise}
  \end{array} \right.
\end{displaymath}
\end{example}

In \fig{fig:example}, we present an example of a channel in this class. Notice that, given that all symbols are "confusable" (i.e. $\prob(Y=0|X=i)>0, \: \forall i\in \textrm{$\cal{X}$}$), the zero-error one-shot capacity of this channel is zero, $\forall n$. However, by allowing a small error probability, we are able to transmit a significant number of bits.

\begin{lemma}
\label{le:oneshot-example}
The $\epsilon$-maximum one-shot capacity of the channel described in \exref{ex:first} satisfies
\vspace{-0.18cm}
\begin{eqnarray}
\label{eq:example}
C_\epsilon = \log (i+1) \textrm{ for $i$ s.t. $e_{i}\leq \epsilon < e_{i+1}$},
\end{eqnarray}
\vspace{-0.55cm}

where $e_0=0$ and $e_n=1$.
\end{lemma}

\begin{proof}
We start by proving that the $\epsilon$-maximum one-shot capacity, $C_\epsilon$, is lower bounded by \eqref{eq:example}. Let $e_{i}\leq \epsilon < e_{i+1}$, for some $i\in \{1,\dots,n-1\}$ (the case $i=0$ is trivial, since by definition $C_\epsilon \geq 0$). Consider the codebook $\underline{\cal{X}}=\{1,\dots,i+1\}$ and the following decoding function:
\begin{displaymath}
\gamma(y) = \left\{ \begin{array}{ll}
 y & \textrm{if $y \in \{1,\dots,i+1\}$}\\ 
 i+1 & \textrm{otherwise}
\end{array} \right.
\end{displaymath}

For $x\in\{1,\dots,i\}$, we have that $\gamma^{-1}(x)=\{x\}$ (where $\gamma^{-1}(x)=\{y\in \textrm{$\cal{Y}$}: \gamma(y)=x \}$) and, thus, $\prob(\gamma(Y)\neq x|X=x)= e_x \leq \epsilon$, because we have that $e_{i}\leq \epsilon < e_{i+1}$ and $e_x\leq e_i$ (because $x\leq i$). With respect to $x=i+1$, we have that $\gamma^{-1}(i+1)=\{0,i+1,\dots,n-1\}$. Moreover, $\prob(Y=i+1|X=i+1)=1-e_{i+1}$ and $\prob(Y=0|X=i+1)=e_{i+1}$. Therefore,\vspace{-0.2cm}$$\prob(\gamma(Y)\neq i+1|X=i+1)=1-\hspace{-0.5cm}\displaystyle\sum_{y\in \gamma^{-1}(i+1)}\hspace{-0.5cm}\prob(Y=y|X=i+1)=0.$$ Hence, we have constructed a pair $(\underline{\cal{X}},\gamma)$ for which $|\underline{\cal{X}}|=i+1$ and $\epsilon_\gamma=\displaystyle\max_{x\in \underline{\cal{X}}}\prob(\gamma(Y)\neq x|X=x)\leq \epsilon$.

Now, we show that $C_\epsilon$ is upper bounded by \eqref{eq:example}. Let $e_{i}\leq \epsilon < e_{i+1}$, for some $i\in \{0,\dots,n-1\}$, and let $(\underline{\cal{X}},\gamma)$ be a pair for which $C_\epsilon=\log |\underline{\cal{X}}|$. Notice that for $x\in \{0\} \cup \{i+1,\dots,n-1\}$, we have $\prob(Y=0|X=x)\geq \epsilon$. Therefore, if $x\in \underline{\cal{X}}\cap (\{0\} \cup \{i+1,\dots,n-1\})$, we must have $0\in \gamma^{-1}(x)$. Thus, since $\gamma$ is a function, we have that $|\underline{\cal{X}}\cap (\{0\} \cup \{i+1,\dots,n-1\})|\leq 1$, which implies that $|\underline{\cal{X}}|\leq i+1$, thus concluding our proof.
\end{proof}

The previous example shows that, by allowing for a small probability of error in the decoding process, we are able to transmit a significantly higher number of bits in one use of the channel, in comparison with the case where no errors are allowed. In the case illustrated in \fig{fig:example}, we have that the $\epsilon$-maximum one-shot capacity verifies

\begin{displaymath}
C_\epsilon = \left\{ \begin{array}{ll}
 0 & \textrm{if $\epsilon<0.01$}\\
 1 & \textrm{if $0.01\leq \epsilon<0.02$}\\
  \log(3) & \textrm{if $\epsilon \geq 0.02$}
  \end{array} \right.
\end{displaymath}

\section{The Case of Maximum Error Probability}
\label{sec:main}

In this section, we present a combinatorial description of the one-shot capacity under maximum error $\epsilon$. We start by defining the graph that will help us obtain the desired description. For that, we first need to use the following definition which associates with each input symbol $x$ a set of output symbols denoted by $D_\epsilon(x)$.

\begin{definition}
For each $x\in \cal{X}$, let \vspace{-0.2cm}$$D_{\epsilon}(x)=\left\{D \subset {\cal{Y}}: \sum_{y \in D} \prob(Y=y|X=x)\geq 1-\epsilon \right\}.$$
\end{definition}

We can view $D_{\epsilon}(x)$ as the set of all possible inverse images of $x$ through a decoding function $\gamma$ (i.e. all possible $\gamma^{-1}(x)$), with $\gamma$ verifying $\epsilon_{\underline{\cal{X}},\gamma}\leq \epsilon$. We are now ready to present the definition of the maximum-one-shot graph of the channel described by $P_{Y|X}$.

\begin{definition}
The {\it maximum-one-shot graph} of the channel described by $P_{Y|X}$ is the graph $G_{\epsilon}$ (with node set $V$ and edge set $E_{\epsilon}$) constructed as follows:
\begin{itemize}
\item the nodes are the elements of the form $(x,D)$ with $x\in \cal{X}$ and $D \in D_{\epsilon}(x)$;
\item two nodes $(x,D)$ and $(x',D')$ are connected if and only if $x=x'$ or $D\cap D' \neq \emptyset$.
\end{itemize}
\end{definition}

\begin{figure}
\centering
\subfigure[]{
  \label{fig:example}
  \includegraphics[width=3.4cm]{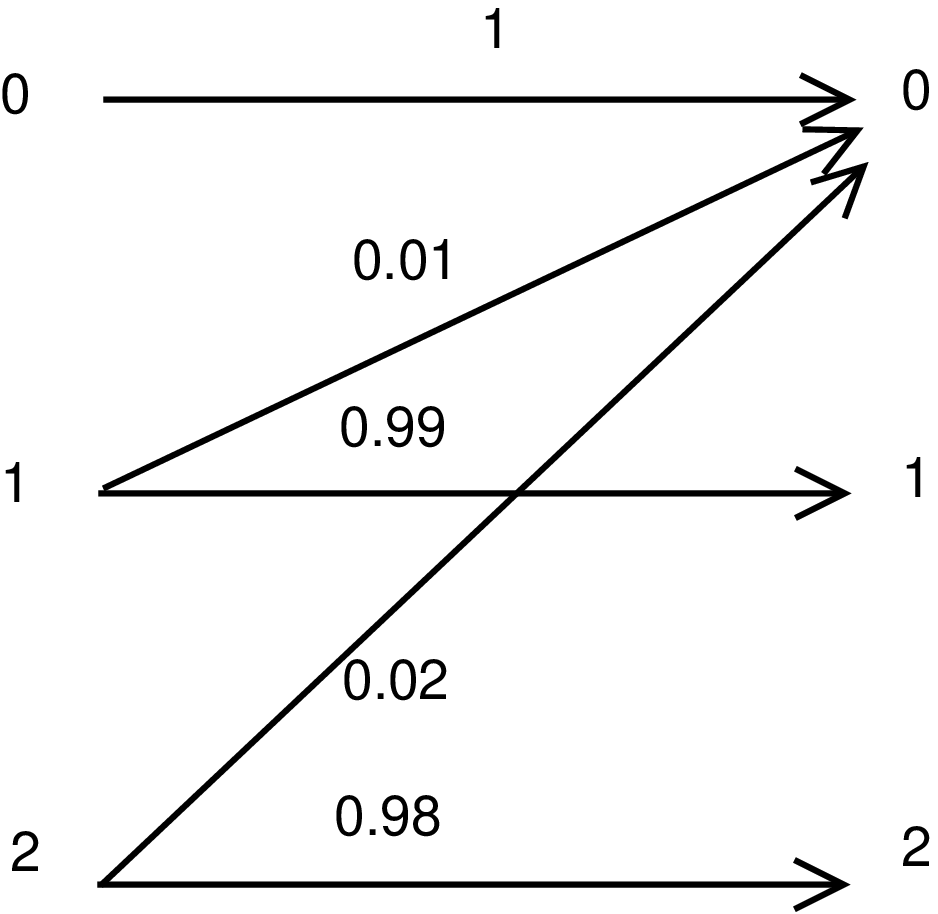}
}
\subfigure[]{
  \label{fig:one_shot_graph}
  \includegraphics[width=4cm]{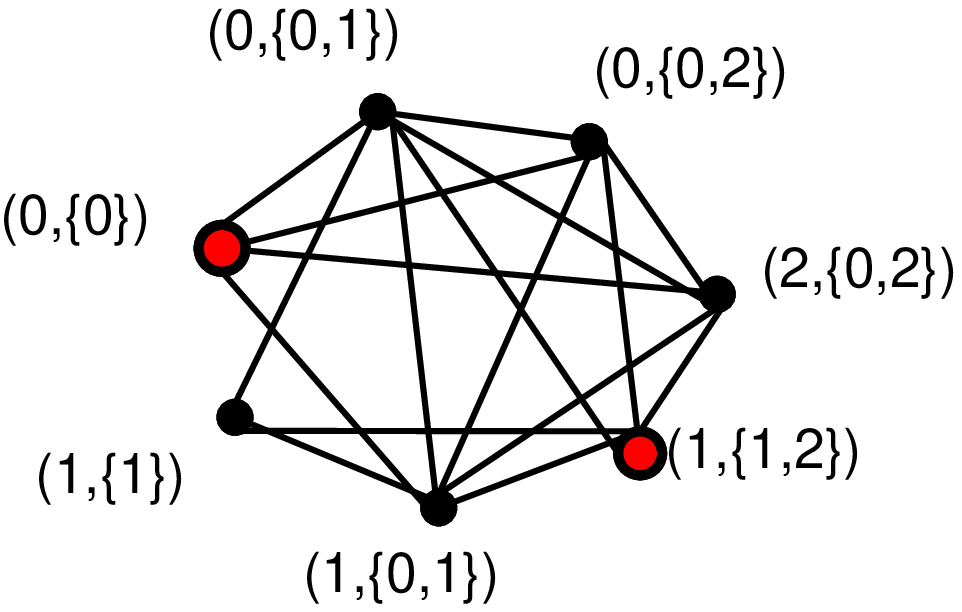}
}
\caption{In (a), we present an instance of the class of channels defined in \exref{ex:first}, with $n=3$, $e_1=0.01$ and $e_2=0.02$. In (b), we present the maximum-one-shot graph of the channel described in (a), for $\epsilon=0.01$. The nodes in red form a maximum independent set, which by \theorref{th:main} implies that the $\epsilon$-maximum one-shot capacity of the channel in (a) is $\log(2)=1$. For the sake of clarity, we excluded the nodes $(0,\{0,1,2\})$, $(1,\{0,1,2\})$ and $(2,\{0,1,2\})$, since these nodes are connected to every other node in the graph and, thus, are not part of any maximum independent set.}
\end{figure}

In \fig{fig:one_shot_graph}, we present the maximum-one-shot graph of the channel in \fig{fig:example}. Due to the definition of $D_\epsilon(x)$, in the maximum-one-shot graph, nodes represent all the possible $\gamma^{-1}(x)$ (such that $\gamma$ is $\epsilon$-admissible). To obtain a proper decoding function from the maximum-one-shot graph, we need to find an {\it independent set}, since a connection between two nodes represents the incompatibility of two inverse images.

\begin{definition}
Consider a graph $G=(V,E)$. An {\it independent set} $I_G$ in $G$ is a set of nodes $v\in V$ in which no two nodes are connected by an edge.
\end{definition}

\begin{definition}
Consider a graph $G=(V,E)$. A {\it maximum independent set} $I_G$ is a largest independent set and its cardinality is called the {\it independence number} of the graph $G$, and it is denoted by $\alpha(G)$.
\end{definition}

Using these definitions, we are now able to state our main result, which relates the one-shot capacity with the independence number of the previously defined graph.

\begin{theorem}
\label{th:main}
Consider a channel described by $P_{Y|X}$ and the corresponding maximum-one-shot graph $G_{\epsilon}=(V,E_{\epsilon})$, with $\epsilon \in [0,1)$. The $\epsilon$-maximum one-shot capacity satisfies \vspace{-0.2cm}$$C_\epsilon=\log(\alpha(G_{\epsilon})).\vspace{-0.2cm}$$
\end{theorem}

We prove this theorem by establishing first that one can transmit a codebook of size at least $\alpha(G_{\epsilon})$ with a single use of the channel. We then show that this is the best one can do.

\begin{lemma}
$C_{\epsilon}\geq \log(\alpha(G_{\epsilon}))$.
\end{lemma}

\begin{proof}
Let $G_{\epsilon}=(V,E_{\epsilon})$ be the maximum-one-shot graph of the channel and let $I_{G_\epsilon}$ be a maximum independent set in $G_\epsilon$. Let $\cal{X}^*$ be the set of symbols in $\cal{X}$ that are represented in $I_{G_\epsilon}$, i.e. \vspace{-0.2cm}$$\textrm{$\cal{X}$}^* =\{x \in \textrm{$\cal{X}$}: \exists D=(y_1,\dots,y_k) \textrm{ such that } (x,D)\in I_{G_\epsilon} \}.\vspace{-0.2cm}$$

For each $x\in \cal{X}^*$, let $d(x)$ be the set of output symbols that are represented in the same node as $x$ in $I_{G_\epsilon}$, i.e. \vspace{-0.2cm}$$d(x)=\{y_1,\dots,y_k: (x,D)\in I_{G_\epsilon} \textrm{ with }D=(y_1,\dots,y_k)\} .\vspace{-0.2cm}$$ Notice that, since $I_{G_\epsilon}$ is an independent set in $G_\epsilon$ and all pairs of nodes of the form $(x,y_1,\dots,y_k)$ and $(x,y'_1,\dots,y'_{k'})$ are connected in $G_\epsilon$, we have that $d(x)$ is unique and properly defined. Let $\textrm{$\cal{Y}$}^*=\{ y\in \cal{Y}:$$ \exists x \in \cal{X}^*$$\textrm{ such that } y\in d(x) \}$.

Now, consider the decoder $\gamma(.)$ constructed as follows:

\begin{itemize}
\item for $y\in \cal{Y}^*$, we set $\gamma(y)=x'$, where $x'\in \cal{X}^*$ is such that $y\in d(x')$;
\item for $y\notin \cal{Y}^*$, we set $\gamma(y)=x^*$, where $x^*$ is some symbol in $\cal{X}^*$.
\end{itemize}

We have that $I_{G_\epsilon}$ is an independent set in $G_\epsilon$. Thus, for every $y\in \cal{Y}^*$, there is only one $x' \in \cal{X}^*$ such that $y\in d(x')$. Therefore, the function $\gamma(.)$ is well-defined, i.e. $\forall y\in {\cal Y},\exists!x\in {\cal X}^*:\gamma(y)=x.$ We also have that $\forall x \in \cal{X}^*,$$\exists y \in \cal{Y}:$$\gamma(y)=x $, which is equivalent to $\gamma(\cal{Y})=\cal{X}^*$. Let $\underline{\cal{X}}=\cal{X}^*$. We have that $\alpha(G_{\epsilon})=|I_{G_\epsilon}|=|\cal{X}^*|$ and, therefore, $|\underline{\cal{X}}|=\alpha(G_{\epsilon})$.

Now, we need to analyze the error probability of the pair $(\underline{\cal{X}},\gamma)$ previously constructed. Let $x \in \underline{\cal{X}}$ and let $\gamma^{-1}(x)=\{y\in \textrm{$\cal{Y}$}: \gamma(y)=x \}=\{y_1,\dots,y_k\}$. We have that \vspace{-0.25cm}$$\prob(\gamma(Y)\neq x|X=x)=1-\displaystyle \sum_{i=1}^{k} \prob(Y=y_i|X=x).\vspace{-0.2cm}$$ Notice that, by the construction of $\gamma(\cdot)$, we have that, for $D=(y_1,\dots,y_k)$, $(x,D)$ is a node of $I_{G_\epsilon}$ and, therefore, a node in $G_\epsilon$. Thus, by the definition of the maximum-one-shot graph $G_\epsilon$, we have that $(y_1,\dots,y_k)\in D_{\epsilon}(x)$, which is equivalent to $\sum_{i=1}^{k} \prob(Y=y_i|X=x)\geq 1-\epsilon.$ Therefore, we have that $\prob(\gamma(Y)\neq x|X=x)\leq \epsilon$, and this inequality is not dependent on the choice of $x \in \underline{\cal{X}}$. Therefore, we have that $\forall x\in \underline{\cal{X}}$, $\prob(\gamma(Y)\neq x|X=x)\leq \epsilon$, which is equivalent to $\epsilon_{\underline{\cal{X}},\gamma}\leq \epsilon.$ Thus, we have constructed a pair $(\underline{\cal{X}},\gamma)$ such that $\epsilon_{\underline{\cal{X}},\gamma}\leq \epsilon$ and $|\underline{\cal{X}}|=\alpha(G_{\epsilon})$, which implies that $C_{\epsilon}\geq \log(\alpha(G_{\epsilon}))$.
\end{proof}

We proved that one can transmit $\alpha(G_{\epsilon})$ symbols with a single use of the channel. Now, we prove that it is not possible to transmit more than that.

\begin{lemma}
$C_{\epsilon}\leq \log(\alpha(G_{\epsilon}))$.
\end{lemma}

\begin{proof}
Let $(\underline{\cal{X}},\gamma)$ be a pair such that $|\underline{\cal{X}}|=2^{C_{\epsilon}}$ and $\epsilon_{\underline{\cal{X}},\gamma}\leq \epsilon$. Let $\gamma^{-1}(x)=\{y\in \textrm{$\cal{Y}$}:\gamma(y)=x \}$. Since $\epsilon_{\underline{\cal{X}},\gamma} \leq \epsilon$, we have that $\forall x\in \underline{\cal{X}}$, $\displaystyle \sum_{i=1}^k \prob(Y=y_i|X=x)\geq 1-\epsilon,$ where $\{y_1,\dots,y_k\}=\gamma^{-1}(x)$. Thus, $(y_1,\dots,y_k)\in D_{\epsilon}(x)$ and, therefore, for $D=(y_1,\dots,y_k)$, $(x,D)$ is a node in the maximum-one-shot graph $G_\epsilon=(V,E)$.

Now, notice that, since $\gamma(.)$ is a function with $\cal{Y}$ as domain, we have that $\forall x_1\neq x_2 \in \underline{\cal{X}}$, $\gamma^{-1}(x_1)\cap \gamma^{-1}(x_2)=\emptyset$. Therefore, the set $I=\{(x,D): x\in \textrm{$\underline{\cal{X}}$ and }  D=\gamma^{-1}(x)\}$ is an independent set in $G_\epsilon$, which implies that $|I|\leq \alpha(G_\epsilon)$. Since $|I|=|\underline{\cal{X}}|$, we have that $|\underline{\cal{X}}|\leq \alpha(G_\epsilon)$ and, therefore, $2^{C_\epsilon}\leq \alpha(G_\epsilon)$.
\end{proof}

\section{Complexity of the Computation of the One-Shot Capacity}
\label{sec:nphard}

Up to now we have shown that the $\epsilon$-error one-shot capacity can be characterized by the independence number of the graph $G_\epsilon$. Computing the independence number is known to be an NP-Hard problem. However, it may be the case that the graphs $G_\epsilon$ we obtain in our reduction are of a {\em simple} nature allowing us to find their independence number efficiently. In what follows we show that this is not the case.

\begin{theorem}
\label{th:nphard} 
The computation of the $\epsilon$-maximum one-shot capacity is NP-Hard for $\epsilon <1/3$.
\end{theorem}

\begin{proof}
We will prove the NP-Hardness of the $\epsilon$-maximum one-shot capacity problem by reducing the independent set problem in $3$-regular graphs (a known NP-Hard problem~\cite{stockmeyer:76}) to an instance of the $\epsilon$-maximum one-shot capacity problem, for $\epsilon<1/3$. The reduction technique is similar to the one used in \cite{lovasz:79}.

Consider a $3$-regular graph $G=(V,E)$. We will construct a communication channel driven from this graph as follows: the input alphabet is $\cal{X}$$=V$, the output alphabet is $\cal{Y}$$=E$ and the transition probability distribution is given by
\vspace{-0.15cm}
\begin{displaymath}
\prob(Y=y|X=x) = \left\{ \begin{array}{ll}
 1/3 & \textrm{if $x$ is an endpoint of $y$}\\
 0 & \textrm{otherwise}
  \end{array} \right.
\end{displaymath}

Notice that $\prob_{Y|X}$ is well-defined, since $G$ is a $3$-regular graph (each node has degree $3$) and, therefore, $\forall x\in \cal{X},$ $\sum_{y\in \cal{Y}}\prob(Y=y|X=x)=1$. For each $x \in \cal{X}$, let $d(x)=\{y \in \textrm{$\cal{Y}$}: \prob(Y=y|X=x)>0\}$. Notice that $|d(x)|=3$ and $\forall y \in  d(x), \prob(Y=y|X=x)=1/3$.

Now, we shall focus on the $\epsilon$-maximum one-shot capacity of the previously constructed channel.
Let $\epsilon <1/3$. Let us now construct the maximum-one-shot graph $G_\epsilon=(V_\epsilon,E_\epsilon)$. The node set is composed of elements of the form $(x,y_1,\dots,y_k)$ such that $\displaystyle\sum_{i=1}^k \prob(Y=y_i|X=x)\geq 1-\epsilon \: (>2/3).$ Two nodes $(x,y_1,\dots,y_k)$ and $(x',y'_1,\dots,y'_{k'})$ are connected in $G_\epsilon$ if and only if $x=x'$ or $\exists i,j$ such that $y_i=y'_j$.
As $\epsilon < 1/3$, notice that for any node $(x,D)$ in $G_\epsilon$ it holds that $d(x) \subseteq D$.

We now show that $\alpha(G_\epsilon)$=$\alpha(G)$.
This suffices to prove our assertion since computing the independence number of $G$ is NP-Hard and the $\epsilon$-maximum one-shot capacity is equal to the (logarithm of the) independence number of $G_\epsilon$.
Namely, we prove that a maximum independent set in $G_\epsilon$ corresponds to a maximum independent set in the original $3$-regular graph $G$, and vice-versa. 

Let $I_\epsilon$ be an independent set in $G_\epsilon$.
Consider the set $I=\{x:\exists D \ \mbox{s.t.}\ (x,D)\in I_\epsilon\}$.
It holds that $|I|=|I_\epsilon|$.
Moreover, for any two nodes $(x,D)$ and $(x',D')$ in $I_\epsilon$ it holds that $d(x) \subseteq D$, $d(x') \subseteq D'$ and $D \cap D' = \phi$.
This implies that $d(x) \cap d(x') = \phi$, which in turn implies that $x$ and $x'$ are not connected by an edge.
We conclude that $I$ is an independent set in $G$.

For the other direction, Let $I$ be an independent set in $G$.
Consider the set $I_\epsilon=\{(x,d(x)):x \in I\}$.
It holds that $|I|=|I_\epsilon|$.
Moreover, for any two nodes $x$ and $x'$ in $I$ it holds that $d(x) \cap d(x')=\phi$.
This implies that $(x,d(x))$ and $(x',d(x'))$ are not connected by an edge in $G_\epsilon$.
We conclude that $I_\epsilon$ is an independent set in $G_\epsilon$.
\end{proof}

\section{The Case of Average Error Probability}
\label{sec:average}

In this section, we devote our attention to the $\epsilon$-average one-shot capacity (\definref{def:average} in \secref{sec:problem}).

\begin{definition}
For each $x\in \cal{X}$, we define \vspace{-0.15cm}$$D(x)=\left\{D \subset {\cal{Y}}: \displaystyle\sum_{y \in D}\prob(Y=y|X=x)>0 \right\}.$$
\end{definition}

We are now ready to present the definition of the {\em average}-one-shot graph of the channel described by $P_{Y|X}$.

\begin{definition}
The {\it average-one-shot graph} of the channel described by $P_{Y|X}$ is the weighted graph $G_\epsilon$ (with node set $V$ and edge set $E_\epsilon$), constructed as follows:
\vspace{-0.2cm}
\begin{itemize}
\item the nodes are the elements of the form $(x,D)$ with $x\in \cal{X}$ and $D \in D(x)$;
\item two nodes $(x,D)$ and $(x',D')$ are connected by an infinite weight edge if and only if $x=x'$ or $D\cap D' \neq \emptyset$.
\item all other pairs of nodes $(x,D)$ and $(x,D')$ are connected by an edge of weight $\prob(Y\not \in D|X=x) +  \prob(Y\not \in D'|X=x').$
\end{itemize}
\end{definition}

The previous definition provides us a tool to describe a relationship between the $\epsilon$-average one-shot capacity and sparse sets in the average-one-shot graph.

\begin{definition}
Consider a weighted graph $G=(V,E)$. An {\it $\epsilon$-sparse set} $I_{\epsilon}$ in $G$ is a set of nodes $v\in V$ for which the weight of edges in the subgraph induced on $I_\epsilon$ is at most $\epsilon |I_\epsilon|(|I_\epsilon|-1)$. For example, a $0$-sparse set is an independent set.
\end{definition}

\begin{definition}
Consider a graph $G=(V,E)$. A {\it maximum $\epsilon$-sparse set} $I_\epsilon$ is a largest set in $G$ that is $\epsilon$-sparse, its cardinality is called the {\it $\epsilon$-sparse number} of the graph $G$, and is denoted by $\alpha_\epsilon(G)$.
\end{definition}

We are now ready to present our main result related to the $\epsilon$-average one-shot capacity.

\begin{theorem}
\label{th:average}
Let $\epsilon \in [0,1]$.
Consider a channel described by $P_{Y|X}$ and let $G_\epsilon=(V,E)$ be the average-one-shot graph. The $\epsilon$-average one-shot capacity is given by $\bar{C}_{\epsilon}=\log(\alpha_\epsilon(G_\epsilon)).$
\end{theorem}

\begin{proof}
Let $G_\epsilon=(V,E)$ be the average-one-shot graph of the channel and let $I_\epsilon$ be a maximum $\epsilon$-sparse set in $G_\epsilon$. Let $\cal{X}^*$ be the set of symbols in $\cal{X}$ that are represented in $I_{\epsilon}$, i.e. $\textrm{$\cal{X}$}^* =\{x \in \textrm{$\cal{X}$}: \exists D \textrm{ such that } (x,D)\in I_{\epsilon} \}.$
Notice that $I_\epsilon$ cannot contain two vertices $(x,D)$ and $(x,D')$ (as they share an edge of infinite weight), or two vertices $(x,D)$ and $(x',D')$ with $D \cap D' \ne \phi$ (for the same reason).

For each $x\in \cal{X}^*$, let $(x,D_x) \in I_\epsilon$. Let $\textrm{$\cal{Y}^*$}=\cup_{x \in \cal{X}^*} D_x$. Now, consider the decoder $\gamma(.)$ constructed as follows:
\vspace{-0.1cm}
\begin{itemize}
\item for $y\in D_x$, we set $\gamma(y)=x$.
\item for $y\notin \cal{Y}^*$, we set $\gamma(y)=x^*$, where $x^*$ is some symbol in $\cal{X}^*$.
\end{itemize}

We have that $I_{\epsilon}$ is an $\epsilon$-sparse set in $G$.
\vspace{-0.2cm} 
{\small{
\begin{eqnarray*}
\epsilon 
&\geq & 
\frac{1}{|I_\epsilon|(|I_\epsilon|-1)}\sum_{x \ne x' \in \cal{X}^*}\prob(Y\not \in D_x|X=x)+\\ &&\qquad \qquad \qquad \qquad \qquad +  \prob(Y\not \in D_{x'}|X=x') \\
&=& 
\frac{|I_\epsilon|-1}{|I_\epsilon|(|I_\epsilon|-1)}\sum_{x \in \cal{X}^*}\prob(Y\not \in D_x|X=x) \\
&=& 
\frac{1}{|I_\epsilon|}\sum_{x \in \cal{X}^*}\prob(\gamma(Y) \ne x|X=x) 
\end{eqnarray*}}}
\vspace{-0.3cm}

This implies that $\bar{C}_{\epsilon}\geq \log{|I_\epsilon|} = \log(\alpha_\epsilon(G))$. For the other direction, let $(\underline{\cal{X}},\gamma)$ be a pair in $\bar{{\cal {A}}_\epsilon}$ such that $\bar{C}_\epsilon = \log{|\underline{\cal{X}}|}$.
Let $I_\epsilon = \{(x,\gamma^{-1}(x))\mid x\in \underline{\cal{X}}\}$.
Clearly, $|I_\epsilon|=|\underline{\cal{X}}|$.
We now show (very similar to the analysis above) that $I_\epsilon$ is $\epsilon$-sparse.
Namely, first notice that $I_\epsilon$ does not contain any infinite weight edges (as $\gamma$ is a decoding function).
Moreover
\vspace{-0.2cm}
{\small{
\begin{eqnarray*}
\epsilon 
&\geq & 
\frac{1}{|\underline{\cal{X}}|}\sum_{x \in \underline{\cal{X}}}\prob(\gamma(Y) \ne x|X=x) \\
&=& 
\frac{1}{|I_\epsilon|}\sum_{x \in I_\epsilon}\prob(\gamma(Y) \ne x|X=x) \\
& =& 
\frac{|I_\epsilon|-1}{|I_\epsilon|(|I_\epsilon -1|)}\sum_{x \in I_\epsilon}\prob(Y\not \in \gamma^{-1}(x)|X=x) \\
&=& 
\frac{1}{|I_\epsilon|(|I_\epsilon|-1)}\sum_{x \ne x' \in I_\epsilon}\prob(Y\not \in D_x|X=x)+\\ &&  \qquad \qquad \qquad \qquad \qquad +  \prob(Y\not \in D_{x'}|X=x')
\end{eqnarray*}}}
\vspace{-0.4cm}

We conclude that $\log{\alpha_\epsilon(G)} \geq \log{|\underline{\cal{X}}|} = \bar{C}_\epsilon$, which concludes our proof.
\end{proof}

\section{Conclusions}
\label{sec:conclu}

Intrigued by the capacity of discrete channels that can be used only once, we elaborated on the $\epsilon$-one-shot capacity, defined as the maximum number of bits that can be transmitted with one channel use while assuring that the decoding error probability is not greater than $\epsilon$. Based on this definition, we introduced the concept of the $\epsilon$-one-shot graph associated with a discrete channel and provided an exact characterization of the $\epsilon$-one-shot capacity through combinatorial properties of the $\epsilon$-one-shot graph. Using this formulation, we prove that computing the $\epsilon$-one-shot capacity (for $\epsilon<1/3$) is NP-Hard.

The practical relevance of the concept we present in this paper was discussed through a non-trivial example of a class of discrete channels for which the zero-error capacity is null, but allowing for small error probability enables the transmission of a significant number of bits in a single use of the channel. 

\section*{Acknowledgment}

This work was partially funded by the Funda\c{c}\~ao para a Ci\^encia e Tecnologia (FCT, Portuguese Foundation for Science and Technology) under grants SFRH-BD-27273-2006 and PTDC/EIA/71362/2006 (WITS project), and by ISF grant 480/08.

\vspace{-0.4cm}
\bibliographystyle{IEEEbib}
\bibliography{pp}

\end{document}